%% file: fullversion.tex
\documentclass[english, a4paper]{article}

\RequirePackage[a4paper, heightrounded, text={14cm,22.2cm}, ignoreall]{geometry}

\usepackage{amsmath}
\usepackage{amsthm}
\usepackage{amssymb}
\usepackage{authblk}
\usepackage{hyperref}
\usepackage{microtype}
\usepackage{url} 
\usepackage{subcaption}
\usepackage[indent, skip=0.25\baselineskip plus 2pt]{parskip}
\usepackage{booktabs}
\usepackage{float}

\newtheoremstyle{mystyle}
  {10pt} 
  {10pt} 
  {\itshape} 
  {} 
  {\bfseries} 
  {.} 
  { } 
  {} 

\theoremstyle{mystyle}

\newtheorem{corollary}{Corollary}
\newtheorem{lemma}{Lemma}
\newtheorem{theorem}{Theorem}

\DeclareRobustCommand{\orcidID}[1]{{\hypersetup{hidelinks}\href{https://orcid.org/#1}{\includegraphics[scale=.03]{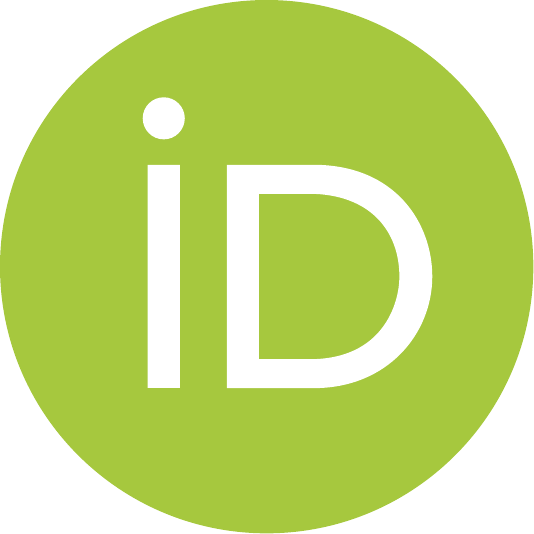}}}%
}

\usepackage{todonotes}

\title{Signotopes Induce Unique Sink Orientations on Grids\footnote{
	S.~Roch was funded by the DFG-Research Training Group 'Facets of Complexity' (DFG-GRK~2434). Special thanks to Rimma Hämäläinen for the stimulating discussions on this topic.
}}

\author{Sandro M. Roch\;\orcidID{0000-0002-9353-9413}}
\affil{Technische Universität Berlin\\
  \texttt{sandro.m.roch@posteo.de}}

\begin{document}

\maketitle

\begin{abstract}
    A unique sink orientation (USO) is an orientation of the edges of a polytope in which every face contains a unique sink. For a product of simplices~$\Delta_{m-1} \times \Delta_{n-1}$, Felsner, Gärtner and Tschirschnitz~(2005) characterize USOs which are induced by linear functions as the USOs on a~$(m \times n)$-grid that correspond to a~\textit{two-colored arrangement of lines}. We generalize some of their results to products~$\Delta^1 \times\cdots\times \Delta^r$ of~$r$ simplices,~USOs on $r$-dimensional grids and~\textit{$(r+1)$-signotopes}.
\end{abstract}

\section{Introduction}
\label{sec:introduction}

In the typical setting of linear programming, one has given a polytope~$P\subset \mathbb{R}^n$ and a linear objective function~\mbox{$f: \mathbb{R}^n\to \mathbb{R}$}. Assuming that~$f$ is sufficiently generic,~$f$ induces an orientation~$v\to w$ of any edge~$\{v, w\}$ of~$P$ s.t.~$f(v) < f(w)$. This is an acyclic orientation of the skeleton of~$P$ with the property that on any face~$F\subseteq P$ it has a unique source and a unique sink. A \textit{source} resp. \textit{sink} on~$F$ is a vertex on which all edges of~$F$ are oriented outgoing resp. incoming. The unique source and the unique sink of~$f$ correspond to the minimum and maximum of~$f$ on~$F$.

A (not necessarily linear) function~$f$ on the vertices of~$P$ that on each face induces a unique source and a unique sink is called an \textit{abstract objective function} by Kalai~\cite{kalai97}. If~$f$ is a linear function, Holt and Klee~\cite{holtKlee99} showed an even stronger property: On any face~$F$, there exist~$\operatorname{dim} F$ many internally disjoint dipaths from the source to the sink. This property is called the \textit{Holt-Klee condition}. 

In~\cite{felsnerGaertnerTschirschnitz05}, Felsner,~Gärtner and~Tschirschnitz treat the special case in which~$P$ is a simple $d$-polytope with~$d+2$ facets. Recall that two polytopes are~\textit{combinatorially equivalent} if their face lattices are isomorphic. A simple~$d$-polytope with~$d+2$ facets is combinatorially equivalent to the Cartesian product~\mbox{$\Delta_{m-1}\times\Delta_{n-1}$} of two simplices, with $m + n = d + 2$; see the example in Figure~\ref{fig:simplex_products}. Every face of~$\Delta_{m-1}\times\Delta_{n-1}$ is the Cartesian product of a face of~$\Delta_{m-1}$ and a face of~$\Delta_{n-1}$. Hence, if~$V$ and~$W$ denote the sets of vertices of~$\Delta_{m-1}$ and $\Delta_{n-1}$, then the~$m\cdot n$ vertices of~$\Delta_{m-1}\times\Delta_{n-1}$ can be identified with~$V\times W$, and for every~$V'\subseteq V$,~$W'\subseteq W$, the vertices~$V'\times W'$ form a face.

\begin{figure}[tb]
    \centering
    \includegraphics{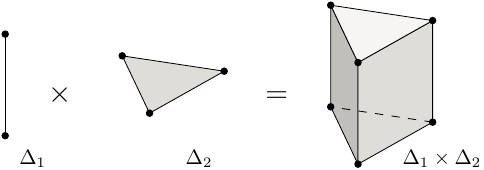}
    \caption{As a Cartesian product of simplices, the $3$-polytope~$\Delta_1\times\Delta_2$ has 5 facets.}
    \label{fig:simplex_products}
\end{figure}

It is because of this simple fact that an orientation of a simple $d$-polytope with~$d+2$ facets can be easily described as an orientation of a two-dimensional grid. For later purposes, we directly define a \textit{grid}~$G$ \textit{of size~$(n_1, \cdots, n_r)$} as the graph on the~\mbox{$r$-tuples}~$[n_1]\times\cdots\times [n_r]$,~$n_i\geq 1$, where two of them form an edge if they differ in exactly one coordinate. Note that in the literature, this is often called a~\textit{rook's graph} instead of a grid. The~\textit{dimension}~$\operatorname{dim} G$ of a grid is the number of~$i$ for which~$n_i > 1$.

For subsets~$V_i\subseteq [n_i]$,~$V_i\neq\emptyset$, the graph induced on~$V_1\times\cdots\times V_r$ is called a \textit{subgrid}. Any face~$F\neq\emptyset$ of~$\Delta_{m-1}\times\Delta_{n-1}$ corresponds uniquely to a subgrid~\mbox{$S=S_1\times S_2\subseteq[m]\times [n]$}. One should not confuse its dimension~$\operatorname{dim} S$ with the polytope dimension~$\operatorname{dim} F$, which is equal to~$\lvert S_1\rvert + \lvert S_2\rvert - 2$. A~\textit{grid orientation} is an orientation of the edges of a grid. A grid orientation naturally induces a grid orientation on any subgrid. We say a grid orientation~$\mathcal{O}$ contains a grid orientation~$\mathcal{O}'$ if there is a subgrid~$S$ on which the orientation induced by~$\mathcal{O}$ is~\textit{isomorphic} to~$\mathcal{O}'$.

Let us make precise what we consider as isomorphic orientations: An edge~$\{v, w\}$ in a grid is~\textit{of dimension~$i$} if~$v$ and~$w$ differ in the $i$-th coordinate. Two grid orientations~$\mathcal{O}$ and~$\mathcal{O}'$ are \textit{isomorphic} if there is a graph isomorphism~\mbox{$\varphi: \mathcal{O}\to \mathcal{O}'$} that preserves orientation, and in which a pair of edges~$\{v, w\}$ and~$\{v', w'\}$ is of the same dimension if and only if~$\{\varphi(v), \varphi(w)\}$ and~$\{\varphi(v'), \varphi(w')\}$ are of the same dimension. Such an isomorphism is a permutation of the~$r$ coordinates combined with a permutation of each of the sets~$[n_i]$.

Given a grid orientation~$\mathcal{O}$ and a subgrid~$S\subseteq\mathcal{O}$, a~\textit{sink} of~$S$ is a vertex~$v\in S$ whose incident edges in~$S$ are all incoming. In the case~$S=\mathcal{O}$, we call~$v$ a~\textit{(global) sink} of~$\mathcal{O}$. A \textit{unique sink (grid) orientation (USO)} is a grid orientation in which each subgrid contains a unique sink. This is equivalent to saying that any subgrid has a unique \textit{source}, where a source is defined analogously to a sink~\cite[Theorem 1]{gaertnerMorrisRuest08}. Although abstract objective functions were previously studied, the term unique sink orientation was coined by Szabo and Welzl~\cite{szaboWelzl01} in~2001. First introduced specifically for hypercubes, they were later generalized for grids~\cite{gaertnerMorrisRuest08}.

A USO is called~\textit{admissible}, if in any subgrid of size~$(n_1', \cdots, n_r')$ there exist~$n_1' + \cdots + n_r' - r$ many internally disjoint dipaths from the unique source to the unique sink\footnote{In difference to~\cite{felsnerGaertnerTschirschnitz05}, our definition of admissibility does not require acyclicity. However, two-dimensional USOs are always acyclic.}. This is exactly the Holt-Klee condition for faces. One can show that a USO of a two-dimensional grid is admissible if and only if it does not contain the~\textit{double-twist}~$\text{DT}$ in Figure~\ref{fig:pla2uso:double_twist}~\cite{felsnerGaertnerTschirschnitz05}. Indeed, the double twist violates the Holt-Klee condition, as there are only two instead of three internally disjoint paths from the source to the sink. Figure~\ref{fig:pla2uso:uso}(a) shows an example of~an~admissible~USO.

\begin{figure}[b]
    \centering
    \includegraphics[page=3]{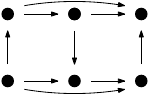}
    \caption{Double twist~$\text{DT}$; acyclic but non-admissible USO}
    \label{fig:pla2uso:double_twist}
\end{figure}

Furthermore, a grid orientation is called~\textit{realizable} if it is the orientation induced by a linear function on a polytope~$P$ which is combinatorially equivalent to a product of simplices. The reason why we only ask for a combinatorially equivalent polytope is this: A precise Cartesian product of simplices contains many classes of parallel edges. For instance, in~$\Delta_1\times\Delta_2$ (see Figure~\ref{fig:simplex_products}), there are four classes of parallel edges. In a linear orientation, parallel edges must have the same orientation. For our analysis, this is too strict. We allow the faces to be perturbed as long as the face lattice is still isomorphic to a product of simplices.

Every realizable orientation is admissible; and every admissible orientation is, by definition, a unique sink orientation. However, for both implications, the reverse is not true. It was one of the original motivations for the study of USOs to characterize realizable orientations. Admissibility is necessary but not sufficient.

For any vertex~$x=(x_1, \cdots, x_r)$ in a grid orientation, let~$\operatorname{rf}_i(x)$ be the out-degree in dimension~$i$, i.e.~the number of neighbors~$y=(y_1, \cdots, y_r)$ with orientation~$x \rightarrow y$ that have~$x_i \neq y_i$ and~$x_j = y_j$ for all~$j\neq i$. The~\textit{refined index} of~$x$ is defined as~\[\operatorname{rf}(x) := (\operatorname{rf}_1(x), \cdots, \operatorname{rf}_r(x)).\] It is easy to see that a USO is uniquely determined by its refined index. As an example, Figure~\ref{fig:pla2uso:uso}(b) shows the refined index of the orientation on the left. Observe that all tuples~\mbox{$[4]_0\times [4]_0$} (where $[n]_0 := \{0, \cdots, n\}$) appear as a refined index exactly once. This is not a coincidence:

\begin{lemma}[Theorem 2 in \cite{gaertnerMorrisRuest08}]\label{lemma:usoIffRifBij}
    The refined index of a USO is a bijection \[\operatorname{rf}: [n_1]\times\cdots\times [n_r]\to [n_1 - 1]_0 \times\cdots\times [n_r - 1]_0\;.\]
\end{lemma}

\begin{figure}[tb]
    \centering
    \includegraphics[page=1]{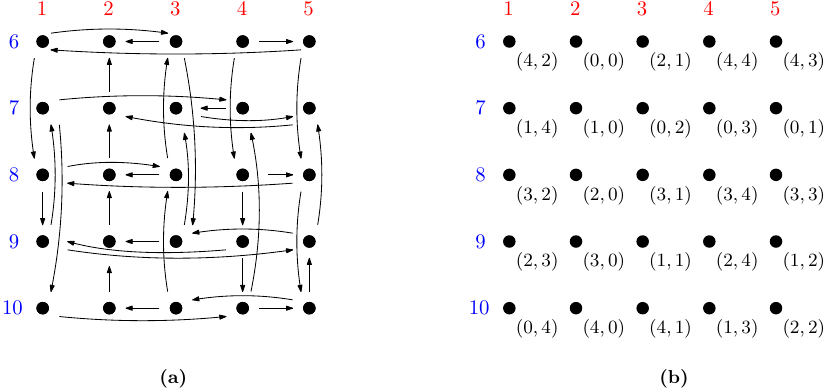}
    \caption{(a): Example of an admissible USO. Transitive arrows are not drawn. (b): Refined index of the orientation in~(a).}
    \label{fig:pla2uso:uso}
\end{figure}

\section{Block Colored Pseudoline Arrangements and Admissible USOs}
\label{pla2uso_known}

An \textit{arrangement of (straight) lines} is a finite family of straight lines~$g_1, \cdots, g_n\subset\mathbb{R}^2$ in the Euclidean plane, no two of which are parallel. As the lines are non-parallel, they pairwise cross in exactly one point. 

A \textit{(Euclidean) arrangement of pseudolines} (or simply \textit{pseudoline arrangement}) is a generalization of an arrangement of lines; it maintains the pairwise crossing property, but lifts the restriction to straight lines. Formally, an~\textit{arrangement of pseudolines} is a finite family of simple curves~$f_1, \cdots, f_n: \mathbb{R}\to\mathbb{R}^2$ with \[\lim_{t\to\infty} \lVert f_i(t)\rVert = \lim_{t\to-\infty} \lVert f_i(t)\rVert = \infty\;,\] called \textit{pseudolines}, which fulfill the property that each two of them intersect at exactly one point where they cross. We label the pseudolines from~$1$ to~$n$ in counterclockwise order as in~Figure~\ref{fig:pla2uso:arrangement}(a). In this article, we only consider \textit{simple} arrangements,~i.e.~at each crossing exactly two pseudolines meet. They are in correspondence with~\textit{$3$-signotopes}. Signotopes were introduced by Manin and Schechtmann~\cite{maninSchechtman89} and further developed by Felsner and Weil~\cite{felsnerWeil01}. They are defined as follows.

For an~$(r+1)$-subset $A\subseteq [n]$ and~$1\leq i \leq n$, let~$A^{\lfloor i\rfloor}$ be the~$r$-subset of~$A$ that contains all elements except the~$i$-th element in increasing order. For example,~\mbox{$\{2, 5, 7, 8\}^{\lfloor2\rfloor} = \{2, 7, 8\}$}. For a map~$\chi: \binom{[n]}{r}\to\{-, +\}$, we define the notation~$\chi(A) := \left(\chi(A^{\lfloor 1\rfloor}), \cdots, \chi(A^{\lfloor r+1\rfloor})\right)$ as the sequence of signs that is obtained by successively omitting an element and applying~$\chi$. 

Now, for~$1\leq r \leq n$, \mbox{an~\textit{$r$-signotope}} is a map~$\chi: \binom{[n]}{r}\to\{-, +\}$ such that for each~\mbox{$(r+1)$-subset}~$A\subseteq [n]$ the sequence~$\chi(A)$ contains at most one sign change. We call this the~\textit{monotonicity property} of signotopes. For example, for~$3$-signotopes, the monotonicity property implies that for each tuple of four elements~$1\leq i < j < k < l \leq n$ we have \begin{align*}
    \left(\chi(jkl), \chi(ikl), \chi(ijl), \chi(ijk)\right)\in\big\{ &(-,-,-,-), (-,-,-,+), (-,-,+,+),\\
    & (-,+,+,+), (+,+,+,+), (+,+,+,-),\\
    &(+,+,-,-), (+,-,-,-))\big\}\; .
\end{align*}

Here, we use a simplified notation~$\chi(ijk) := \chi\left(\{i,j,k\}\right)$, which we will use from time to time for convenience.

In a pseudoline arrangement, every triple of pseudolines~$1\leq i < j < k\leq n$ forms exactly one of the two configurations shown in Figure~\ref{fig:triple_orientations}. The triple orientations of pseudoline arrangements can be precisely characterized as~$3$-signotopes (cf. \cite[Theorem 7]{felsnerWeil01}).

\begin{figure}[tbh]
    \centering
    \includegraphics[page=1]{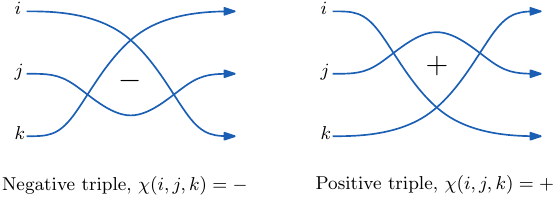}
    \caption{Each triple of pseudolines~$i<j<k$ is either negative or positive.}
    \label{fig:triple_orientations}
\end{figure}

A~\textit{block colored pseudoline arrangement} is a simple arrangement of~$n$ pseudolines that consists of two consecutive blocks of~$r$ red and~$b$ blue pseudolines, $n=r+b$. We assume the pseudolines to be numbered as usual with~$1, \cdots, n$ in counterclockwise order, starting with the red pseudolines~$1, \cdots, r$ and continuing with the blue pseudolines~\mbox{$r+1, \cdots, n$}; see the example in Figure~\ref{fig:pla2uso:arrangement}(a).

It will be convenient to consider two block colored pseudoline arrangements~$\mathcal{A}$ and~$\mathcal{A}'$ with the same number of red and number of blue pseudolines as isomorphic, if they form the same arrangement up to monochromatic triangle flips. This notion of isomorphism captures the information in which order a blue (red) pseudoline intersects the red (blue) pseudolines, and also the order in which it intersects a pair of a red and a blue pseudoline.

The following theorem is one of the main results of~\cite{felsnerGaertnerTschirschnitz05}.

\begin{theorem}[Felsner, Gärtner \& Tschirschnitz~\cite{felsnerGaertnerTschirschnitz05}]\label{thm:felsnerGaertnerTschirschnitz}
    Any two-dimensional USO is admissible if and only if it is induced by a block colored pseudoline arrangement. Moreover, it is realizable if and only if it is induced by an arrangement of straight lines. 
\end{theorem}

In the following, we will explain what is meant in Theorem~\ref{thm:felsnerGaertnerTschirschnitz} by saying that an arrangement induces an admissible orientation, sketch its proof, and formulate a precise bijection.

\subsection{Pseudoline arrangements induce admissible orientations}

Assume that we have given a block colored arrangement~$\mathcal{A}$ and want to assign to it a unique sink orientation of size~$(b, r)$. Let~$p$ be a red and~$q$ a blue pseudoline. With~$x$ being the number of blue pseudolines that~$p$ crosses before crossing~$q$, and~$y$ being the number of red pseudolines that~$q$ crosses before crossing~$p$, we call~$(x, y)$ the \textit{crossing index}~$\operatorname{cri}_\mathcal{A}(p, q)$.

It has been shown in~\cite{felsnerGaertnerTschirschnitz05} that~$\mathcal{A}$ can always be drawn as a \textit{grid drawing}; see the example in~Figure~$\ref{fig:pla2uso:arrangement}$: There is a~$(b \times r)$-grid of points where the red pseudolines pass through from top to bottom as vertically monotonic curves, while the blue pseudolines pass through from left to right as horizontally monotonic curves, and each crossing between red and blue pseudolines lies on a grid point. The existence of a grid drawing can also be deduced from a result that we will give later (Corollary~\ref{cor:refined_index_is_bijection}).

Note that in a grid drawing, the~$r\cdot b$ red-blue crossings must be bijectively distributed over the~$r\cdot b$ grid points, so that each red pseudoline~$p$ contains exactly~$b$ of them and each blue pseudoline~$q$ exactly~$r$. Therefore, the monotonicity implies that the crossing of~$p$ with~$q$ must be placed at grid point~$\operatorname{cri}_\mathcal{A}(p, q)$ (we index the grid points starting with~$(0, 0)$ from the top left corner).

So the crossing index of~$\mathcal{A}$ can be used to construct the grid drawing. But it also determines a unique sink orientation: Write down the crossing indices table-wise for all~$p$ and~$q$ and interpret them as refined indices of a unique sink orientation. The example in~Figure~\ref{fig:pla2uso:uso} is the unique sink orientation obtained from the arrangement in Figure~\ref{fig:pla2uso:arrangement}(a). For example, the crossing index of the first red and the first blue pseudoline is~$\operatorname{cri}_\mathcal{A}(1,6) = (4, 2)$, so we get~$\operatorname{rf}(1,1) = (4, 2)$.

One can check that for inducing a double twist~$\text{DT}$, a pair of pseudolines has to cross twice, which is forbidden in a pseudoline arrangement. Therefore, unique sink orientations induced by pseudoline arrangements are always admissible.

\begin{figure}[tb]
    \centering
    \includegraphics[page=2]{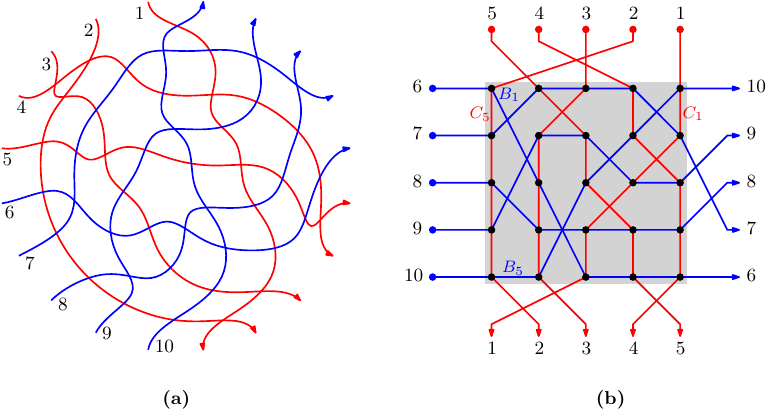}
	
    \caption{(a):~Block colored pseudoline arrangement~$\mathcal{A}$ of~$r=5$ red and~$b=5$ blue pseudolines. (b):~Grid drawing of~$\mathcal{A}$: Every crossing between red and blue pseudolines lies on a point of~a~$(5\times 5)$-grid.}
    \label{fig:pla2uso:arrangement}
\end{figure}

\subsection{Admissible orientations induce pseudoline arrangements}

Given an admissible unique sink orientation~$\mathcal{O}$ of size~$(b, r)$, we can construct a block colored arrangement~$\mathcal{A}$ which induces~$\mathcal{O}$ as described in the previous subsection. The orientation~$\mathcal{O}$ is determined by its refined index~$\operatorname{rf}_\mathcal{O}$. For each row~$1\leq i\leq b$, the map~\mbox{$j \mapsto\left(\operatorname{rf}_\mathcal{O}(i, j) \right)_2$} is bijective, as it is the refined index of the unique sink orientation induced on the one-dimensional subgrid~$\{i\}\times [r]$~(Lemma~\ref{lemma:usoIffRifBij}). Therefore, if we interpret the tuples~$\left\{ \operatorname{rf}_\mathcal{O}(i, 1), \cdots, \operatorname{rf}_\mathcal{O}(i, r) \right\}$ as grid coordinates, we have exactly one point in each column and can interpolate a horizontally monotonic blue curve through them. Analogously, every~$1\leq j\leq r$ defines a vertically monotonic red curve that passes through the grid. The bijectivity of~$\operatorname{rf}_\mathcal{O}$ further implies that each grid point is contained in exactly one blue and one red curve. We refer the reader to~\cite{felsnerGaertnerTschirschnitz05} for the proof of the fact that the admissibility of the orientation implies that each two blue and each two red curves cross at most once.

So far we have that an admissible orientation~$\mathcal{O}$ determines a family of red and blue curves within the convex hull of the grid points (gray box in Figure~\ref{fig:pla2uso:arrangement}(b)); each blue curve crosses each red curve exactly once, and curves of the same color cross at most once. Not much is missing for a pseudoline arrangement; clearly, one can complete the missing crossings outside of the grid. However, this completion is not unique; it depends on how we identify the curves with the pseudolines~$1,\cdots, r+b$.

Let us give an example. Call~$B_i$ the blue curve that contains the grid point~$(i - 1, 0)$, and~$R_j$ the red curve that contains the point~$(0, r - j)$. In order to obtain the arrangement in Figure~\ref{fig:pla2uso:arrangement}, we have to identify
\begin{align*}
    (B_1, B_2, B_3, B_4, B_5) &= (6, 7, 8, 9, 10)\\
    (R_1, R_2, R_3, R_4, R_5) &= (1, 4, 3, 5, 2)
\end{align*}
and connect the curves to the labels~$1, \cdots, 10$ that must appear in counterclockwise order around the grid.

Let~$\mathcal{B} := \{B_1, \cdots, B_b\}$ be the set of blue curves, and~$\mathcal{R} := \{R_1, \cdots, R_r\}$ be the set of red curves, both sets indexed as above. The freedom that we have when identifying~$\mathcal{B}\cup\mathcal{R}$ with the pseudolines~$1,\cdots, r+b$ can be expressed in terms of partial orders on the sets~$\mathcal{B}$ and~$\mathcal{R}$. For $i < j$, say~$B_i \prec_\mathcal{B} B_j$ if~$B_i$ and~$B_j$ cross, and~$R_i \prec_\mathcal{R} R_j$ if~$R_i$ and~$R_j$ cross. Clearly, these partial orders depend on the orientation~$\mathcal{O}$, as the entire construction of the curves~$\mathcal{B}$ and~$\mathcal{R}$ depends on~$\mathcal{O}$.

\begin{theorem}\label{thm:felsner_gaertner_tschirschnitz_bij}
    There is a bijection between block colored arrangements of~$r$ red and~$b$ blue pseudolines and admissible orientations of a grid of size~$(b, r)$ together with a pair of linear extensions of~$(\mathcal{B}, \prec_\mathcal{B})$ and~$(\mathcal{R}, \prec_\mathcal{C})$.
\end{theorem}
\begin{proof}
    We have seen that the admissible orientations of a~$(b, r)$-grid are in bijection with families~$(\mathcal{B}, \mathcal{R})$ of blue paths $\mathcal{B} := \{B_1, \cdots, B_b\}$ and red paths~$\mathcal{R} := \{R_1, \cdots, R_r\}$ within the grid. We claim that the linear extensions of~$(\mathcal{B}, \prec_\mathcal{B})$ and~$(\mathcal{R}, \prec_\mathcal{R})$ correspond exactly to the valid ways of identifying~$\mathcal{B}$ and~$\mathcal{R}$ with the blue and red pseudolines.
    
    If, for~$i < j$, $B_i$ and $B_j$ do not intersect, then, regardless of which blue pseudolines~$p_i$,~$p_j$ we identify them with, eventually~$p_i$ will cross~$p_j$ exactly once. This crossing will lie to the right of the grid if~$p_i < p_j$, and to the left of the grid otherwise. However, if~$B_i$ and~$B_j$ cross, we have to identify them with pseudolines $p_i < p_j$, as otherwise, if $p_j < p_i$, the pseudolines $p_i$ and $p_j$ would cross twice. Hence, an identification of~$\mathcal{B}$ with the blue pseudolines is valid if and only if it is a linear extension of~$(\mathcal{B}, \prec_\mathcal{B})$.

    Moreover, an arrangement determines~$(B_1, \cdots, B_b) = (p_1, \cdots, p_b)$ uniquely: If~$B_i$ and~$B_j$ do not intersect, swopping~$p_i$ and~$p_j$ causes a blue-blue crossing to move from the left side of the grid to the right side of the grid, or the other way around. This cannot be achieved by only flipping monochromatic triangles, so this corresponds to a different block colored arrangement. 
    
    Analogous arguments hold for the identification of the red curves.
\end{proof}

\section{Block Signotopes and Unique Sink Orientations}
\label{pla2uso_generalized}

Before, we described a way to obtain a USO~$\mathcal{O}$ from a block colored arrangement~$\mathcal{A}$. Let~$\chi_\mathcal{A}$ be the $3$-signotope corresponding to~$\mathcal{A}$. Observe that we can read off the orientation~$\mathcal{O}$ directly from~$\chi_\mathcal{A}$: Say, pseudolines~$i, j$ are red and pseudoline~$k$ is blue,~$i < j < k$. If~$\chi_\mathcal{A}(i, j, k) = +$, pseudoline~$k$ first crosses~$i$ before crossing~$j$, hence,~$\left(\operatorname{cri}_\mathcal{A}(i, k)\right)_2 < \left(\operatorname{cri}_\mathcal{A}(j, k)\right)_2$, and in the resulting grid orientation~$\mathcal{O}$ we must have~$(k, i) \leftarrow (k, j)$. So, a sign of triples consisting of two elements of one color and one of the other color determines the orientation of exactly one edge. This is precisely what we aim to generalize for grids of arbitrary dimension~$r$ and~$(r+1)$-signotopes.

We call a partition $[n] = C_1 \;\dot{\cup}\;\cdots\;\dot{\cup}\; C_r$ a~\textit{block partition}, if~$c < c'$ for all $c\in C_i, c'\in C_j$ whenever $i < j$. For example,~$(C_1, C_2, C_3) = (\{1, 2\}, \{3, 4, 5\}, \{6, 7\})$ is a block partition of~$[7]$. A signotope $\chi: \binom{[n]}{r+1}\to\{-, +\}$ of rank $r+1$ together with a block partition~\mbox{$[n] = C_1 \;\dot{\cup}\;\cdots\;\dot{\cup}\; C_r$} is called a~\textit{block signotope} of~\textit{rank}~$(r+1)$. The previous subsection treated the special case~$r=2$; the two classes were the classes of red and blue pseudolines. We are now going to generalize this to a greater number of classes.

Let $G$ be the grid on~$C_1\times\cdots\times C_r$. The block signotope~$\chi$ naturally induces a grid orientation~$O_\chi$ on $G$: Suppose $v, v'\in C_1\times\cdots\times C_r$ differ exactly in the $i$-th component, so~\mbox{$v=(c_1, \cdots, c_{i-1}, c, c_{i+1}, \cdots, c_r)$} and $v'=(c_1, \cdots, c_{i-1}, c', c_{i+1}, \cdots, c_r)$ with~\mbox{$c, c'\in C_i$}, and assume $c < c'$. If \mbox{$\chi(c_1, \cdots, c_{i-1}, c, c', c_{i+1}, \cdots, c_r)=+$}, then, in~$\mathcal{O}_\chi$, orient $v\to v'$; otherwise, orient $v'\to v$. Note that, in the case~$r=3$, compared to the construction in the previous section,~$\mathcal{O}_\chi$ is now exactly reversed and transposed. Indeed, if~$i < j < k$ and~$\chi(i, j, k) = +$, then, in~$\mathcal{O}_\chi$, we have~$(i, k) \rightarrow (j, k)$.

\begin{theorem}\label{thm:block_signotopes_uso}
    For every block signotope~$\chi: \binom{[n]}{r+1}\to\{-, +\}$ with block partition~$[n] = C_1 \;\dot{\cup}\;\cdots\;\dot{\cup}\; C_r$,~$\mathcal{O}_\chi$ is a unique sink orientation.
\end{theorem}

In an arrangement of pseudolines~$\mathcal{A}$, the~\textit{arrangement graph}~$G_\mathcal{A}$ that represents crossings as vertices and arcs as directed edges (see the example in Figure~\ref{fig:acyclicity}) is always acyclic; \mbox{see~\cite[Lemma 1]{felsnerWeil01}}. 

\begin{figure}[tbh]
    \centering
    \includegraphics{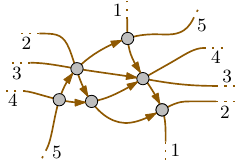}
    \caption{The arrangement graph~$G_\mathcal{A}$ is acyclic.}
    \label{fig:acyclicity}
\end{figure}

The first ingredient for the proof of Theorem~\ref{thm:block_signotopes_uso} is a generalization of the acyclicity of (a supergraph of)~$G_\mathcal{A}$ to signotopes of any rank~\cite[Lemma 10]{felsnerWeil01}. It states that the following graph~$G_\chi$ on~$r$-subsets of~$[n]$ is acyclic: Two vertices~$R, R'\in\binom{[n]}{r}$ share an edge if~$\lvert R\cap R'\rvert = r - 1$, and its orientation is from the lexicographic smaller to the lexicographic larger set, if~$\chi(R\cup R')=+$, and reversed if~$\chi(R\cup R')=-$. The orientation~$\mathcal{O}_\chi$ is identical to the subgraph of~$G_\chi$ induced by the vertices~$\left\{R\subseteq [n]\;:\;\lvert R\cap C_i\rvert = 1\text{ for all }i\right\}$ and thus inherits acyclicity.

\begin{lemma}\label{lemma:signotopes_are_acyclic}
    The orientation $\mathcal{O}_\chi$ is acyclic.
\end{lemma}

This implies that the orientation induced by~$\chi$ on any subgrid contains at least one sink. Next, we will complement this with the fact that any subgrid contains at most one sink. We show this first for at most two-dimensional subgrids and later extend it to general subgrids.

\begin{lemma}\label{lemma:unique_sink_dim2}
    The orientation induced by~$\mathcal{O}_\chi$ on any subgrid $S$ of dimension at most two contains at most one sink.
\end{lemma}

Before giving the proof of Lemma~\ref{lemma:unique_sink_dim2}, we show the following auxiliary statement. It states that for grid orientations, the property of being induced by a signotope inherits to subgrids. For a reader who is guided by intuition, this might be clear, so they may skip the proof.

\begin{lemma}\label{lemma:subgrids_inherit_be_induced}
    Let~$\mathcal{O}_\chi$ be an $r$-dimensional orientation induced by a block signotope~$\chi$ of rank~$(r+1)$, and~$\mathcal{O}'\subseteq\mathcal{O}_\chi$ the orientation induced on an $r'$-dimensional subgrid. Then,~$\mathcal{O}'$ is induced by a block signotope of rank~$(r'+1)$.
\end{lemma}
\begin{proof}
    Let~$[n] = C_1 \;\dot{\cup}\;\cdots\;\dot{\cup}\; C_r$ be the block partition of~$\chi$, so~$\mathcal{O}_\chi$ is an orientation of the grid~$C_1\times\cdots\times C_r$. Let~$C_1'\times\cdots\times C_r'$,~$C_i'\subseteq C_i$ be the subgrid that corresponds to~$\mathcal{O}'$.

    Let~$\chi'$ be the restriction of~$\chi$ on~$(r+1)$-subsets of~$C_1'\;\dot{\cup}\;\cdots\;\dot{\cup}\; C_r'$. Then, identifying the elements~$C_1'\;\dot{\cup}\;\cdots\;\dot{\cup}\; C_r'$ with~$\left[\;\lvert C_1'\rvert + \cdots + \lvert C_r'\rvert\;\right]$, the restriction~$\chi'$ is \mbox{an~$(r+1)$-signotope} that induces~$\mathcal{O}'$ as an orientation of the grid~$C_1'\times\cdots\times C_r'$. If~$r'=r$,~i.e.~$\mathcal{O}'$ has the same dimension as~$\mathcal{O}_\chi$, this is all that needs to be shown.
    
    If the dimension of~$\mathcal{O}'$ is smaller than the dimension of~$\mathcal{O}_\chi$, so~$r' < r$, then there are~$r-r'$ indices~$i$ with~$\lvert C_i'\rvert = 1$. For simplicity, assume that these are the last~$r-r'$ indices, so~$C_i' = \{ c_i \}$ for all~$r'+1\leq i\leq r$, for some~$c_i\in C_i$. We have to show the existence of an~$(r'+1)$-signotope that induces~$\mathcal{O}'$ as an orientation of the grid~$C_1'\times\cdots\times C_{r'}'$. Consider the map\begin{align*}
        \chi'_{\downarrow\{c_{r'+1},\cdots, c_r\}}: \binom{C_1'\cup\cdots\cup C_{r'}'}{r'+1}&\to\{-, +\}\\
        M &\mapsto \chi'(M\cup\{c_{r'+1},\cdots, c_r\})\; .
    \end{align*}
    This is called a~\textit{contraction} of~$\chi'$ and as such it is indeed an~$(r'+1)$-signotope; see~\cite{felsnerWeil01} for details. One verifies that~$\chi'_{\downarrow\{c_{r'+1},\cdots, c_r\}}$ induces~$\mathcal{O}'$ as an orientation of the grid~\mbox{$C_1'\times\cdots\times C_{r'}'$}.
\end{proof}

\begin{proof}[Proof of Lemma \ref{lemma:unique_sink_dim2}]
    Suppose towards a contradiction that in the orientation induced by~$\mathcal{O}_\chi$ on~$S$ there are two sinks,~$c\in S$ and~$c'\in S$,~$c\neq c'$. Clearly, if~$c=(c_i)$ and~$c'=(c_i')$ differ only in a single coordinate,~$c$ and~$c'$ are connected by an edge and cannot be both sinks. We may assume that~$c$ and~$c'$ differ in exactly the first two coordinates. We consider the two-dimensional subgrid~$Q_2\subseteq S$ of four vertices that is spanned by~$c$ and~$c'$. In the orientation induced on~$Q_2$,~$c$ and~$c'$ are again sinks. Omitting the coordinates in which~$c$ and~$c'$ are equal, we identify \[Q_2 = \left\{(c_1, c_2), (c_1, c_2'), (c_1', c_2), (c_1', c_2')\right\}\; .\] 
    
    By Lemma~\ref{lemma:subgrids_inherit_be_induced}, the orientation induced by~$\mathcal{O}_\mathcal{\chi}$ on~$Q_2$ is induced by a block signotope~$\chi'$ of rank~$3$ with block partition~$\{c_1, c_1'\}\;\dot{\cup}\;\{c_2, c_2'\}$. Using the following case distinction (visualized in Figure~\ref{fig:proof_sign_ind_uso}) we show that, assuming~$c$ is a sink,~$c'$ must have at least one outgoing edge, which is a contradiction. 
    
    Recall that, by definition, the fact that~$\{c_1, c_1'\}\;\dot{\cup}\;\{c_2, c_2'\}$ is a block partition implies~$\{c_1, c_1'\} < \{c_2, c_2'\}$ (pairwise). This leaves the following four cases.

    \begin{itemize}
        \item Case~$c_1 < c_1' < c_2 < c_2'$: 
        
        As~$c$ is a sink, we must have~$\chi(c_1, c_2, c_2') = -$ and~$\chi(c_1, c_1', c_2) = -$. The monotonicity property implies~$\chi(c_1, c_1', c_2') = -$.

        \item Case~$c_1 < c_1' < c_2' < c_2$: 
        
        As~$c$ is a sink, we must have~$\chi'(c_1, c_2', c_2) = +$ and~\mbox{$\chi(c_1, c_1', c_2) = -$}. The monotonicity property implies~$\chi'(c_1, c_1', c_2') = -$.

        \item Case~$c_1' < c_1 < c_2 < c_2'$: 
        
        As~$c$ is a sink, we must have~$\chi'(c_1, c_2, c_2') = -$ and~\mbox{$\chi'(c_1', c_1, c_2) = +$}. The monotonicity property implies~$\chi'(c_1', c_2, c_2') = -$ or~$\chi'(c_1', c_1, c_2') = +$.

        \item Case~$c_1' < c_1 < c_2' < c_2$: 
        
        As~$c$ is a sink, we must have~$\chi'(c_1, c_2', c_2) = +$ and~\mbox{$\chi'(c_1', c_1, c_2) = +$}. The monotonicity property implies~$\chi'(c_1', c_2', c_2) = +$.
    \end{itemize}

    \begin{figure}[ht]
        \centering
        \includegraphics{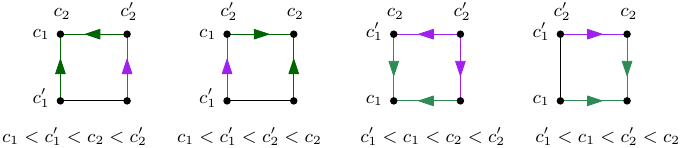}
        \caption{In each case, the two green arcs imply at least one of the purple arcs.}
        \label{fig:proof_sign_ind_uso}
    \end{figure}
\end{proof}

In order to generalize~Lemma~\ref{lemma:unique_sink_dim2} to higher dimensions, one can make use of the following Theorem due to Joswig, Kaibel and Körner (2002):

\begin{theorem}[Lemma 6 in \cite{joswigKaibelKoerner}]\label{thm:joswig}
    Let~$P$ be a simple~$d$-polytope and~$2\leq k\leq d-1$. If~$\overrightarrow{G_P}$ is an acyclic orientation of the skeleton graph~$G_P$ of~$P$ that has more than one global sink, then there is a~$k$-face of~$P$ on which the orientation induced by~$\overrightarrow{G_P}$ contains more than one sink.
\end{theorem}

\vspace{0.5em}

\begin{lemma}\label{lemma:unique_sink}
    The orientation induced by~$\mathcal{O}_\chi$ on any subgrid $S = C_1'\times\cdots\times C_r'$, $C_i'\subseteq C_i$ contains at most one sink.
\end{lemma}

The argument for deducing Lemma~\ref{lemma:unique_sink} from Lemma~\ref{lemma:unique_sink_dim2} using Theorem~\ref{thm:joswig} is as follows. Suppose that the orientation induced on~$S$ contains two sinks. The grid orientation on~$S$ is an acyclic orientation of the skeleton of a polytope that is combinatorially equivalent to a product of simplices. This is a simple polytope. Hence, by~\mbox{Theorem}~\ref{thm:joswig}, there exists a two-dimensional face having two sinks. This corresponds to~a~\mbox{$(2, 2)$-subgrid} of~$S$ having two sinks, which contradicts Lemma~\ref{lemma:unique_sink_dim2}. 

However, in order to be self-contained, we give an independent proof of Lemma~\ref{lemma:unique_sink}.

\begin{proof}
    We show the statement by induction on the dimension~$k$ of the subgrid~$S$.

    For~$k\in\{1, 2\}$, the statement holds according to~Lemma~\ref{lemma:unique_sink_dim2}, so assume~$k > 2$.
    
    Suppose towards a contradiction, the orientation induced by~$\mathcal{O}_\chi$ on~$S$ contains two sinks~$c, c'\in S$, $c\neq c'$. 
    
    For two vertices~$v, v'\in S$ let $v\;\triangle\; v' := \{ i\; :\; v_i \neq v_i'\}$ denote the coordinates in which~$v$ and~$v'$ differ. We may assume~$\lvert c\;\triangle\; c'\rvert = k$, as otherwise, if~$\lvert c\;\triangle\; c'\rvert < k$, then~$c$ and~$c'$ are contained in a subgrid~$S'\subset S$ of lower dimension, and by induction we know that the orientation on~$S'$ induced by~$\mathcal{O}_\chi$ contains at most one sink. Me may further assume without loss of generality that~$c\;\triangle\; c' = \{1, \cdots, k\}$, i.e.~$c$ and~$c'$ differ in exactly the first~$k$ coordinates.

    Consider the~$k$-dimensional hypercube~$Q_k \subseteq S$ that is spanned by~$c$ and~$c'$, i.e.~$Q_k$ is the minimal subgrid~$S'\subseteq S$ with~$c, c'\in S'$. Omitting the coordinates in which~$c=(c_i)$ and~$c'=(c_i')$ equal, we identify~$Q_k=\{c_1, c_1'\}\times\cdots\times\{c_k, c_k'\}$. The vertices~$c$ and~$c'$ form opposite corners of~$Q_k$, and in the orientation induced by~$\mathcal{O}_\chi$ they are both sinks.

    For $0\leq i \leq k$, define the \textit{layer}~$V_i := \{ v \in Q_k\;:\;\lvert c\;\triangle\; v\rvert = i\}$. Observe that~$V_0 = \{c\}$ and~$V_k=\{c'\}$. We claim that for~$1 < i < k$, every~$v\in V_i$ has a directed edge~$(v, v')$ to some~$v'\in V_{i-1}$. Otherwise, all of the~$i$ neighbors of~$v$ in~$V_{i-1}$ are incoming neighbors. Consider the~$i$-dimensional hypercube~$Q'\subset Q_k$ that is spanned by~$c$ and~$v$. In the orientation induced on~$Q'$, both~$c$ and~$v$ are sinks. But by induction, in this subgrid of dimension less than~$k$ there is at most one sink; contradiction.

    Using a symmetric argument, using that~$c'$ is a sink, we get that for~$1 < i < k$, every~$v\in V_i$ has a directed edge~$(v, v')$ to some~$v'\in V_{i+1}$. 
    
    Then, for any pair of successive layers~$V_i, V_{i+1}$ with~$1 \leq i, i+1 \leq k-1$ we have that each~$v\in V_i$ has an outgoing edge to~$V_{i+1}$, and each~$v\in V_{i+1}$ has an outgoing edge to~$V_i$. Following these edges, one can find a directed circuit within~$V_i\cup V_{i+1}$. However, by Lemma~\ref{lemma:signotopes_are_acyclic}, the orientation~$\mathcal{O}_\chi$ is acyclic, so is the orientation induced on~$Q_k$; contradiction.
\end{proof}

\begin{proof}[Proof of Theorem \ref{thm:block_signotopes_uso}]
    Consider the orientation induced by~$\mathcal{O}_\chi$ on any subgrid. Lemma~\ref{lemma:signotopes_are_acyclic} shows acyclicity, which implies that there exists at least one sink, while Lemma~\ref{lemma:unique_sink} shows that there exists at most one sink. Hence, there is exactly one sink.
\end{proof}

Consider a block signotope~$\chi$ of rank~$(r+1)$ with block partition~$[n] = C_1\;\dot{\cup}\;\cdots\;\dot{\cup}\; C_r$. For an element $(c_1, \cdots, c_r)\in C_1 \times \cdots \times C_r$ we define the \textit{index in block $i$} as \begin{align*}
    \operatorname{rf}_i(c_1, \cdots, c_r) := \hspace{0.5cm} & \#\left\{ c\in C_i : c > c_i,\hspace{0.1cm}\chi(c_1, \cdots, c_i, c, c_{i+1}, \cdots, c_r) = + \right\} \\
    +\hspace{0.1cm} &\#\left\{ c\in C_i : c < c_i,\hspace{0.1cm}\chi(c_1, \cdots, c_{i-1}, c, c_i, \cdots, c_r) = - \right\}\hspace{0.1cm},
\end{align*}
and its \textit{refined index} as $$\operatorname{rf}_{\chi}(c_1, \cdots, c_r) := \left(\operatorname{rf}_1(c_1, \cdots, c_r), \cdots, \operatorname{rf}_r(c_1, \cdots, c_r)\right)\hspace{0.1cm}.$$ In terms of~$\mathcal{O}_\chi$, the index~$\operatorname{rf}_i(c_1, \cdots, c_r)$ is the out-degree of vertex~$(c_1, \cdots, c_r)$ in the~$i$-th dimension. Thus,~$\operatorname{rf}_\chi = \operatorname{rf}_{\mathcal{O}_\chi}$ coincide.

Clearly, $\operatorname{rf}_\chi(c_1, \cdots, c_r)\in \{0, \cdots, \lvert C_1\rvert - 1\}\times\cdots\times\{0, \cdots, \lvert C_r\rvert - 1\}$. The following result is an immediate consequence of Theorem~\ref{thm:block_signotopes_uso}.

\begin{corollary}\label{cor:refined_index_is_bijection}
    For any $(r+1)$-signotope~$\chi$ and block partition~$[n] = C_1\;\dot{\cup}\;\cdots\;\dot{\cup}\; C_r$, the refined index is a bijection $$\operatorname{rf}_\chi: C_1\times\cdots\times C_r\to\{0, \cdots, \lvert C_1\rvert - 1\}\times\cdots\times\{0, \cdots, \lvert C_r\rvert - 1\}\hspace{0.1cm}.$$  
\end{corollary}
\begin{proof}
    By Theorem~\ref{thm:block_signotopes_uso},~$\mathcal{O}_\chi$ is a unique sink orientation. By Lemma~\ref{lemma:usoIffRifBij}, its refined index is bijective. It coincides with~$\operatorname{rf}_\chi$.
\end{proof}

We consider Corollary~\ref{cor:refined_index_is_bijection} to be a result of independent interest: As it holds for any block partition~\mbox{$[n] = C_1 \;\dot{\cup}\;\cdots\;\dot{\cup}\; C_r$}, it provides general insight into the structure of signotopes.

\subsection{Admissibility of~\texorpdfstring{$3$-dimensional}{3-dimensional} USOs induced by signotopes}

As mentioned earlier, in the case of two-dimensional acyclic unique sink orientations, admissibility is equivalent to avoiding the double twist~$\text{DT}$. Moreover, as a further characterization, an acyclic unique sink orientation is admissible if and only if it is induced by an arrangement of pseudolines. Does this generalize to higher-dimensional grids?

Gärtner~\cite{gaertner02} has determined that, up to isomorphy, there exist only two acyclic unique sink orientations of a three-dimensional cube that are not admissible. We call them~$\text{NAC}_1$ and~$\text{NAC}_2$; they are both shown in Figure~\ref{fig:non_admissible_cubes}. Furthermore, he showed that these together with the double twist~$\text{DT}$ as forbidden patterns completely characterize admissibility of acyclic unique sink orientations up to dimension three:

\begin{figure}[tb]
    \centering
    \includegraphics[page=1]{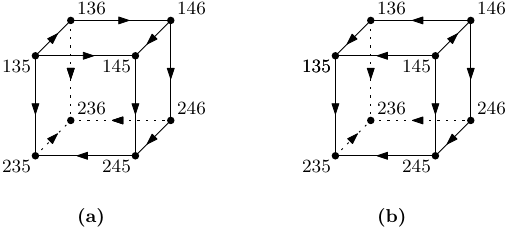}

    \caption{Two non-admissible cube orientations.\;\;(a):~$\text{NAC}_1$ (b):~$\text{NAC}_2$}
    \label{fig:non_admissible_cubes}
\end{figure}

\begin{theorem}[Theorem 4.1 in \cite{gaertner02}]\label{thm:characterization_non_admissible}
    An acyclic unique sink orientation of dimension of at most~$3$ is admissible if and only if it does not contain~$\text{NAC}_1$,~$\text{NAC}_2$ and~$\text{DT}$.
\end{theorem}

Our following theorem shows that one direction of~Theorem~\ref{thm:felsnerGaertnerTschirschnitz} can be generalized from $3$-signotopes alias pseudoline arrangements to~$4$-signotopes: they also induce admissible grid orientations.

\begin{theorem}\label{thm:3d_usos_from_signo_admissible}
    Every unique sink orientation~$\mathcal{O}_\chi$ on a grid of dimension at most~$3$ that is induced by a block signotope~$\chi$ is admissible.
\end{theorem}

\begin{proof}
    By Lemma~\ref{lemma:signotopes_are_acyclic} we know that~$\mathcal{O}_\chi$ is acyclic. Using Theorem~\ref{thm:characterization_non_admissible}, it is sufficient to show that~$\mathcal{O}_\chi$ does not contain any of~$\text{NAC}_1$, $\text{NAC}_2$ and~$\text{DT}$. 
    
    Suppose,~$\mathcal{O}_\chi$ contains~$\text{NAC}_1$, $\text{NAC}_2$ or~$\text{DT}$. Recall that this means that for some subgrid~$S\subseteq\mathcal{O}_\chi$, the orientation induced on~$S$ is isomorphic to~$\text{NAC}_1$, $\text{NAC}_2$ or~$\text{DT}$. Hence, according to Lemma~\ref{lemma:subgrids_inherit_be_induced}, there is an orientation isomorphic to~$\text{NAC}_1$, $\text{NAC}_2$ or~$\text{DT}$ which is induced by a signotope. In the following, we will show that this is impossible,  which proves the theorem.

    For the double twist~$\text{DT}$, it is easy to see that it is not induced by a~$3$-signotope: the corresponding arrangement would contain a pair of pseudolines that cross twice; see~\cite{felsnerGaertnerTschirschnitz05} for details.

    Let~$S_2\wr S_3$ denote the subgroup of~$S_6$ that is generated by~$(12)$,~$(34)$,~$(56)$,~$(13)(24)$ and~$(24)(35)$. The group~$S_2\wr S_3$ is isomorphic to the \textit{hyperoctahedral group} which consists of the~$48$ symmetries of a three-dimensional cube. For each of~$\text{NAC}_1$ and~$\text{NAC}_2$, there are~$3!\cdot 2^3 = 48$ orientations isomorphic to it. We must show that none of them is induced by~a~$4$-signotope. Alternatively, we apply any of the~$48$ permutations~$S_2\wr S_3$ on the coordinates as in Figure~\ref{fig:non_admissible_cubes} and show that the obtained orientation is not induced by a signotope.
    
    For example, Figure~\ref{fig:non_admissible_cube_permuted} shows the orientation obtained from~$\text{NAC}_1$~(Figure~\ref{fig:non_admissible_cubes}(a)) by permuting the coordinates by the permutation~$\pi=(5,6,2,1,4,3)$. It cannot be induced by a signotope, because the orientation of the three red edges would imply two sign changes in~$\chi(12345)$\footnote{Recall the notation from signotopes: For an~$r$-signotope~$\chi$ and an~$(r+1)$-subset~$A\subseteq [n]$, 
    the sequence of signs~$\chi(A) := \left(\chi(A^{\lfloor 1\rfloor}), \cdots, \chi(A^{\lfloor r+1\rfloor})\right)$ contains at most one sign change.}.
    
    Using a computer program, for each~$\pi\in S_2\wr S_3$ we found a $5$-tuple for which in the corresponding sign sequence imposed by the orientation there are two sign changes; see our data in Appendix~\ref{appendix_sec:pl:3dusos_from_signo_admissible}.

    \begin{figure}[htb]
        \centering
        \includegraphics[page=3]{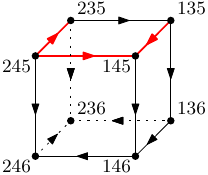}
        \caption{Non-admissible cube orientation~$\text{NAC}_1$ with labels permuted by~$\pi = (5, 6, 2, 1, 4, 3)$. The red edges impose~$\chi(2345)=-$,~$\chi(1345)=+$ and~$\chi(1245)=-$. This contradicts the monotonicity property of a signotope.}
        \label{fig:non_admissible_cube_permuted}
    \end{figure}
\end{proof}

However, the converse of Theorem~\ref{thm:3d_usos_from_signo_admissible} is not true: Figure~\ref{fig:admissible_cube_not_induced} shows an admissible orientation of a cube that is not induced by a signotope. This fact was verified again by computer assistance using the same method that we used in the proof of Theorem~\ref{thm:3d_usos_from_signo_admissible} in order to show that~$\text{NAC}_1$ and~$\text{NAC}_2$ are not induced by signotopes. The data can be found in Appendix~\ref{appendix_sec:pl:admissible_uso_not_induced}.

\begin{figure}[H]
    \centering
    \includegraphics[page=4]{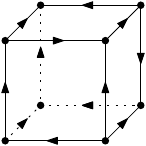}
    \caption{Admissible cube orientation that is not induced by a signotope.}
    \label{fig:admissible_cube_not_induced}
\end{figure}

\section{Conclusion and Future Work}
\label{sec:conclusion}

Our mission was to generalize the work from~\cite{felsnerGaertnerTschirschnitz05} to higher dimensions. For this purpose, we defined block signotopes and demonstrated in Theorem~\ref{thm:block_signotopes_uso} that they induce a USO of a grid of corresponding dimension. Moreover, for dimension three they are admissible (Theorem~\ref{thm:3d_usos_from_signo_admissible}). However, unlike in the two-dimensional case, not all higher-dimensional admissible USOs are induced by block signotopes. We conclude with the following question: Is the USO~$\mathcal{O}_\chi$ always admissible for dimensions greater than three?

\bibliography{bibliography}

\input{fullversion_appendix}

\end{document}

%% file: fullversion_appendix.tex
\appendix
\newpage

\section{Data for Proof of Theorem \ref{thm:3d_usos_from_signo_admissible}}
\label{appendix_sec:pl:3dusos_from_signo_admissible}

\begin{table}[h] 
    \centering
    \fontsize{6}{7} \selectfont
    \begin{minipage}{0.45\textwidth}
    \begin{tabular}{c|c}
        $\pi$ & violating sequence \\
        \midrule
        $(1, 2, 3, 4, 5, 6)$ & $\chi(23456) = *-+--$ \\
        $(1, 2, 3, 4, 6, 5)$ & $\chi(13456) = *+-++$ \\
        $(1, 2, 4, 3, 5, 6)$ & $\chi(12345) = +-++*$ \\
        $(1, 2, 4, 3, 6, 5)$ & $\chi(12345) = +-++*$ \\
        $(1, 2, 5, 6, 3, 4)$ & $\chi(13456) = *++-+$ \\
        $(1, 2, 6, 5, 3, 4)$ & $\chi(12356) = +-*++$ \\
        $(1, 2, 5, 6, 4, 3)$ & $\chi(23456) = *--+-$ \\
        $(1, 2, 6, 5, 4, 3)$ & $\chi(12356) = +-*++$ \\
        $(2, 1, 3, 4, 5, 6)$ & $\chi(13456) = *-+--$ \\
        $(2, 1, 3, 4, 6, 5)$ & $\chi(23456) = *+-++$ \\
        $(2, 1, 4, 3, 5, 6)$ & $\chi(12345) = -+--*$ \\
        $(2, 1, 4, 3, 6, 5)$ & $\chi(12345) = -+--*$ \\
        $(2, 1, 5, 6, 3, 4)$ & $\chi(23456) = *++-+$ \\
        $(2, 1, 6, 5, 3, 4)$ & $\chi(12356) = -+*--$ \\
        $(2, 1, 5, 6, 4, 3)$ & $\chi(13456) = *--+-$ \\
        $(2, 1, 6, 5, 4, 3)$ & $\chi(12356) = -+*--$ \\
        $(3, 4, 1, 2, 5, 6)$ & $\chi(12345) = ++-+*$ \\
        $(3, 4, 1, 2, 6, 5)$ & $\chi(12345) = ++-+*$ \\
        $(4, 3, 1, 2, 5, 6)$ & $\chi(12345) = --+-*$ \\
        $(4, 3, 1, 2, 6, 5)$ & $\chi(12345) = --+-*$ \\
        $(5, 6, 1, 2, 3, 4)$ & $\chi(12346) = -+--*$ \\
        $(6, 5, 1, 2, 3, 4)$ & $\chi(12345) = -+--*$ \\
        $(5, 6, 1, 2, 4, 3)$ & $\chi(12345) = +-++*$ \\
        $(6, 5, 1, 2, 4, 3)$ & $\chi(12346) = +-++*$
    \end{tabular}
    \end{minipage}
    \hspace{1cm}
    \begin{minipage}{0.45\textwidth}
    \begin{tabular}{c|c}
        $\pi$ & violating sequence \\
        \midrule
        $(3, 4, 2, 1, 5, 6)$ & $\chi(12456) = +-*++$ \\
        $(3, 4, 2, 1, 6, 5)$ & $\chi(12356) = -+*--$ \\
        $(4, 3, 2, 1, 5, 6)$ & $\chi(12356) = +-*++$ \\
        $(4, 3, 2, 1, 6, 5)$ & $\chi(12456) = -+*--$ \\
        $(5, 6, 2, 1, 3, 4)$ & $\chi(12346) = +-++*$ \\
        $(6, 5, 2, 1, 3, 4)$ & $\chi(12345) = +-++*$ \\
        $(5, 6, 2, 1, 4, 3)$ & $\chi(12345) = -+--*$ \\
        $(6, 5, 2, 1, 4, 3)$ & $\chi(12346) = -+--*$ \\
        $(3, 4, 5, 6, 1, 2)$ & $\chi(12356) = ++*-+$ \\
        $(3, 4, 6, 5, 1, 2)$ & $\chi(12356) = --*+-$ \\
        $(4, 3, 5, 6, 1, 2)$ & $\chi(12456) = ++*-+$ \\
        $(4, 3, 6, 5, 1, 2)$ & $\chi(12456) = --*+-$ \\
        $(5, 6, 3, 4, 1, 2)$ & $\chi(12345) = ++-+*$ \\
        $(6, 5, 3, 4, 1, 2)$ & $\chi(12346) = ++-+*$ \\
        $(5, 6, 4, 3, 1, 2)$ & $\chi(12345) = --+-*$ \\
        $(6, 5, 4, 3, 1, 2)$ & $\chi(12346) = --+-*$ \\
        $(3, 4, 5, 6, 2, 1)$ & $\chi(12456) = --*+-$ \\
        $(3, 4, 6, 5, 2, 1)$ & $\chi(12456) = ++*-+$ \\
        $(4, 3, 5, 6, 2, 1)$ & $\chi(12356) = --*+-$ \\
        $(4, 3, 6, 5, 2, 1)$ & $\chi(12356) = ++*-+$ \\
        $(5, 6, 3, 4, 2, 1)$ & $\chi(12346) = --+-*$ \\
        $(6, 5, 3, 4, 2, 1)$ & $\chi(12345) = --+-*$ \\
        $(5, 6, 4, 3, 2, 1)$ & $\chi(12346) = ++-+*$ \\
        $(6, 5, 4, 3, 2, 1)$ & $\chi(12345) = ++-+*$
    \end{tabular}
    \end{minipage}
    \vspace{10pt}
    \caption{For each permutation~$\pi\in S_2\wr S_3$ of~$\text{NAC}_1$, there is a $5$-tuple whose corresponding sign sequence contains two sign changes.}
    \medskip
    \label{app:tbl:pl:violated_sequences_first_cube}
\end{table}
\begin{table}[H]
    \centering
    \fontsize{6}{7} \selectfont
    \begin{minipage}{0.45\textwidth}
    \begin{tabular}{c|c}
        $\pi$ & violating sequence \\
        \midrule
        $(1, 2, 3, 4, 5, 6)$ & $\chi(12356) = +-*++$ \\
        $(1, 2, 3, 4, 6, 5)$ & $\chi(12456) = +-*++$ \\
        $(1, 2, 4, 3, 5, 6)$ & $\chi(12456) = +-*++$ \\
        $(1, 2, 4, 3, 6, 5)$ & $\chi(12356) = +-*++$ \\
        $(1, 2, 5, 6, 3, 4)$ & $\chi(12345) = +-++*$ \\
        $(1, 2, 6, 5, 3, 4)$ & $\chi(12346) = +-++*$ \\
        $(1, 2, 5, 6, 4, 3)$ & $\chi(12346) = +-++*$ \\
        $(1, 2, 6, 5, 4, 3)$ & $\chi(12345) = +-++*$ \\
        $(2, 1, 3, 4, 5, 6)$ & $\chi(12356) = -+*--$ \\
        $(2, 1, 3, 4, 6, 5)$ & $\chi(12456) = -+*--$ \\
        $(2, 1, 4, 3, 5, 6)$ & $\chi(12456) = -+*--$ \\
        $(2, 1, 4, 3, 6, 5)$ & $\chi(12356) = -+*--$ \\
        $(2, 1, 5, 6, 3, 4)$ & $\chi(12345) = -+--*$ \\
        $(2, 1, 6, 5, 3, 4)$ & $\chi(12346) = -+--*$ \\
        $(2, 1, 5, 6, 4, 3)$ & $\chi(12346) = -+--*$ \\
        $(2, 1, 6, 5, 4, 3)$ & $\chi(12345) = -+--*$ \\
        $(3, 4, 1, 2, 5, 6)$ & $\chi(12456) = -+*--$ \\
        $(3, 4, 1, 2, 6, 5)$ & $\chi(12356) = -+*--$ \\
        $(4, 3, 1, 2, 5, 6)$ & $\chi(12356) = -+*--$ \\
        $(4, 3, 1, 2, 6, 5)$ & $\chi(12456) = -+*--$ \\
        $(5, 6, 1, 2, 3, 4)$ & $\chi(12346) = -+--*$ \\
        $(6, 5, 1, 2, 3, 4)$ & $\chi(12345) = -+--*$ \\
        $(5, 6, 1, 2, 4, 3)$ & $\chi(12345) = -+--*$ \\
        $(6, 5, 1, 2, 4, 3)$ & $\chi(12346) = -+--*$
    \end{tabular}
    \end{minipage}
    \hspace{1cm}
    \begin{minipage}{0.45\textwidth}
    \begin{tabular}{c|c}
        $\pi$ & violating sequence \\
        \midrule
        $(3, 4, 2, 1, 5, 6)$ & $\chi(12456) = +-*++$ \\
        $(3, 4, 2, 1, 6, 5)$ & $\chi(12356) = +-*++$ \\
        $(4, 3, 2, 1, 5, 6)$ & $\chi(12356) = +-*++$ \\
        $(4, 3, 2, 1, 6, 5)$ & $\chi(12456) = +-*++$ \\
        $(5, 6, 2, 1, 3, 4)$ & $\chi(12346) = +-++*$ \\
        $(6, 5, 2, 1, 3, 4)$ & $\chi(12345) = +-++*$ \\
        $(5, 6, 2, 1, 4, 3)$ & $\chi(12345) = +-++*$ \\
        $(6, 5, 2, 1, 4, 3)$ & $\chi(12346) = +-++*$ \\
        $(3, 4, 5, 6, 1, 2)$ & $\chi(12346) = ++-+*$ \\
        $(3, 4, 6, 5, 1, 2)$ & $\chi(12345) = ++-+*$ \\
        $(4, 3, 5, 6, 1, 2)$ & $\chi(12346) = --+-*$ \\
        $(4, 3, 6, 5, 1, 2)$ & $\chi(12345) = --+-*$ \\
        $(5, 6, 3, 4, 1, 2)$ & $\chi(12345) = --+-*$ \\
        $(6, 5, 3, 4, 1, 2)$ & $\chi(12346) = --+-*$ \\
        $(5, 6, 4, 3, 1, 2)$ & $\chi(12345) = ++-+*$ \\
        $(6, 5, 4, 3, 1, 2)$ & $\chi(12346) = ++-+*$ \\
        $(3, 4, 5, 6, 2, 1)$ & $\chi(12345) = ++-+*$ \\
        $(3, 4, 6, 5, 2, 1)$ & $\chi(12346) = ++-+*$ \\
        $(4, 3, 5, 6, 2, 1)$ & $\chi(12345) = --+-*$ \\
        $(4, 3, 6, 5, 2, 1)$ & $\chi(12346) = --+-*$ \\
        $(5, 6, 3, 4, 2, 1)$ & $\chi(12346) = --+-*$ \\
        $(6, 5, 3, 4, 2, 1)$ & $\chi(12345) = --+-*$ \\
        $(5, 6, 4, 3, 2, 1)$ & $\chi(12346) = ++-+*$ \\
        $(6, 5, 4, 3, 2, 1)$ & $\chi(12345) = ++-+*$
    \end{tabular}
    \end{minipage}
    \vspace{10pt}
    \caption{For each permutation~$\pi\in S_2\wr S_3$ of~$\text{NAC}_2$, there is a $5$-tuple whose corresponding sign sequence contains two sign changes.}
    \medskip
    \label{app:tbl:pl:violated_sequences_second_cube}
\end{table}

We permute~$\text{NAC}_1$ (resp.~$\text{NAC}_2$) by each~$\pi\in S_2\wr S_3$ and obtain~$48$ isomorphic unique sink orientations. Table~\ref{app:tbl:pl:violated_sequences_first_cube} and~Table~\ref{app:tbl:pl:violated_sequences_second_cube} list for each~$\pi\in S_2\wr S_3$ a~$5$-tuple which witnesses that the obtained orientation is not induced by a~$4$-signotope. If it was, then, in the corresponding sign sequence of length~$5$, there had to be two sign changes, which contradicts the monotonicity property of a signotope. Some of the positions in these sequences are filled by a~"$*$". This indicates that for the corresponding~\mbox{$4$-tuple} the sign is not determined by the cube orientation. Still, independent of the sign at this position, the sequence contains at least two sign changes. 

\section{Data corresponding to Figure \ref{fig:admissible_cube_not_induced}}
\label{appendix_sec:pl:admissible_uso_not_induced}

\begin{table}[h]
    \centering
    \fontsize{6}{7} \selectfont
    \begin{minipage}{0.45\textwidth}
    \begin{tabular}{c|c}
        $(1, 2, 3, 4, 5, 6)$ & $\chi(12345) = -+--*$ \\
        $(1, 2, 3, 4, 6, 5)$ & $\chi(12345) = --+-*$ \\
        $(1, 2, 4, 3, 5, 6)$ & $\chi(12346) = ++-+*$ \\
        $(1, 2, 4, 3, 6, 5)$ & $\chi(12345) = ++-+*$ \\
        $(1, 2, 5, 6, 3, 4)$ & $\chi(12356) = -+*--$ \\
        $(1, 2, 6, 5, 3, 4)$ & $\chi(12456) = ++*-+$ \\
        $(1, 2, 5, 6, 4, 3)$ & $\chi(12356) = --*+-$ \\
        $(1, 2, 6, 5, 4, 3)$ & $\chi(12356) = ++*-+$ \\
        $(2, 1, 3, 4, 5, 6)$ & $\chi(12345) = +-++*$ \\
        $(2, 1, 3, 4, 6, 5)$ & $\chi(12346) = +-++*$ \\
        $(2, 1, 4, 3, 5, 6)$ & $\chi(12356) = ++*-+$ \\
        $(2, 1, 4, 3, 6, 5)$ & $\chi(12356) = --*+-$ \\
        $(2, 1, 5, 6, 3, 4)$ & $\chi(12346) = ++-+*$ \\
        $(2, 1, 6, 5, 3, 4)$ & $\chi(12345) = ++-+*$ \\
        $(2, 1, 5, 6, 4, 3)$ & $\chi(12346) = --+-*$ \\
        $(2, 1, 6, 5, 4, 3)$ & $\chi(12345) = --+-*$ \\
        $(3, 4, 1, 2, 5, 6)$ & $\chi(12356) = ++*-+$ \\
        $(3, 4, 1, 2, 6, 5)$ & $\chi(12356) = --*+-$ \\
        $(4, 3, 1, 2, 5, 6)$ & $\chi(12346) = -+--*$ \\
        $(4, 3, 1, 2, 6, 5)$ & $\chi(12345) = -+--*$ \\
        $(5, 6, 1, 2, 3, 4)$ & $\chi(12345) = ++-+*$ \\
        $(6, 5, 1, 2, 3, 4)$ & $\chi(12346) = ++-+*$ \\
        $(5, 6, 1, 2, 4, 3)$ & $\chi(12345) = --+-*$ \\
        $(6, 5, 1, 2, 4, 3)$ & $\chi(12346) = --+-*$
    \end{tabular}
    \end{minipage}
    \hspace{1cm}
    \begin{minipage}{0.45\textwidth}
    \begin{tabular}{c|c}
        $(3, 4, 2, 1, 5, 6)$ & $\chi(12345) = --+-*$ \\
        $(3, 4, 2, 1, 6, 5)$ & $\chi(12346) = --+-*$ \\
        $(4, 3, 2, 1, 5, 6)$ & $\chi(12345) = ++-+*$ \\
        $(4, 3, 2, 1, 6, 5)$ & $\chi(12345) = +-++*$ \\
        $(5, 6, 2, 1, 3, 4)$ & $\chi(12356) = --*+-$ \\
        $(6, 5, 2, 1, 3, 4)$ & $\chi(12356) = ++*-+$ \\
        $(5, 6, 2, 1, 4, 3)$ & $\chi(12456) = --*+-$ \\
        $(6, 5, 2, 1, 4, 3)$ & $\chi(12356) = +-*++$ \\
        $(3, 4, 5, 6, 1, 2)$ & $\chi(12346) = +-++*$ \\
        $(3, 4, 6, 5, 1, 2)$ & $\chi(12345) = +-++*$ \\
        $(4, 3, 5, 6, 1, 2)$ & $\chi(13456) = *+-++$ \\
        $(4, 3, 6, 5, 1, 2)$ & $\chi(12456) = +-*++$ \\
        $(5, 6, 3, 4, 1, 2)$ & $\chi(12456) = +-*++$ \\
        $(6, 5, 3, 4, 1, 2)$ & $\chi(23456) = *-+--$ \\
        $(5, 6, 4, 3, 1, 2)$ & $\chi(12345) = +-++*$ \\
        $(6, 5, 4, 3, 1, 2)$ & $\chi(12346) = +-++*$ \\
        $(3, 4, 5, 6, 2, 1)$ & $\chi(12346) = -+--*$ \\
        $(3, 4, 6, 5, 2, 1)$ & $\chi(12345) = -+--*$ \\
        $(4, 3, 5, 6, 2, 1)$ & $\chi(23456) = *+-++$ \\
        $(4, 3, 6, 5, 2, 1)$ & $\chi(12456) = -+*--$ \\
        $(5, 6, 3, 4, 2, 1)$ & $\chi(12456) = -+*--$ \\
        $(6, 5, 3, 4, 2, 1)$ & $\chi(13456) = *-+--$ \\
        $(5, 6, 4, 3, 2, 1)$ & $\chi(12345) = -+--*$ \\
        $(6, 5, 4, 3, 2, 1)$ & $\chi(12346) = -+--*$
    \end{tabular}
    \end{minipage}
    \vspace{10pt}
    \caption{For each permutation~$\pi\in S_2\wr S_3$ of the cube orientation in Figure~\ref{fig:admissible_cube_not_induced}, there is a $5$-tuple whose corresponding sign sequence contains two sign changes.}
    \medskip
    \label{app:tbl:pl:violated_sequences_admissible_cube}
\end{table}

%% file: fullversion.bbl
\def\authornoop#1{}
\begin{thebibliography}{1}

\bibitem{felsnerGaertnerTschirschnitz05}
Stefan Felsner, Bernd G{\"a}rtner, and Falk Tschirschnitz.
\newblock Grid orientations, {{\((d,d+2)\)}}-polytopes, and arrangements of
  pseudolines.
\newblock {\em Discrete Comput. Geom.}, 34(3):411--437, 2005.
\newblock \href {https://doi.org/10.1007/s00454-005-1187-x}
  {\path{doi:10.1007/s00454-005-1187-x}}.

\bibitem{felsnerWeil01}
Stefan Felsner and Helmut Weil.
\newblock Sweeps, arrangements and signotopes.
\newblock {\em Discrete Appl. Math.}, 109(1-2):67--94, 2001.
\newblock \href {https://doi.org/10.1016/S0166-218X(00)00232-8}
  {\path{doi:10.1016/S0166-218X(00)00232-8}}.

\bibitem{gaertner02}
Bernd G{\"a}rtner.
\newblock The random-facet simplex algorithm on combinatorial cubes.
\newblock {\em Random Struct. Algorithms}, 20(3):353--381, 2002.
\newblock \href {https://doi.org/10.1002/rsa.10034}
  {\path{doi:10.1002/rsa.10034}}.

\bibitem{gaertnerMorrisRuest08}
Bernd G{\"a}rtner, Walter~D. Morris~jun., and Leo R{\"u}st.
\newblock Unique sink orientations of grids.
\newblock {\em Algorithmica}, 51(2):200--235, 2008.
\newblock \href {https://doi.org/10.1007/s00453-007-9090-x}
  {\path{doi:10.1007/s00453-007-9090-x}}.

\bibitem{holtKlee99}
Fred Holt and Victor Klee.
\newblock A proof of the strict monotone 4-step conjecture.
\newblock In {\em Advances in Discrete and Computational Geometry}, volume 223
  of {\em Contemp. Math.}, pages 201--216. AMS, 1999.

\bibitem{joswigKaibelKoerner}
Michael Joswig, Volker Kaibel, and Friederike K{\"o}rner.
\newblock On the {{\(k\)}}-systems of a simple polytope.
\newblock {\em Isr. J. Math.}, 129:109--117, 2002.
\newblock \href {https://doi.org/10.1007/BF02773157}
  {\path{doi:10.1007/BF02773157}}.

\bibitem{kalai97}
Gil Kalai.
\newblock Linear programming, the simplex algorithm and simple polytopes.
\newblock {\em Math. Program.}, 79(1-3 (B)):217--233, 1997.
\newblock \href {https://doi.org/10.1007/BF02614318}
  {\path{doi:10.1007/BF02614318}}.

\bibitem{maninSchechtman89}
Yu.~I. Manin and V.~V. Shekhtman.
\newblock Arrangements of hyperplanes, higher braid groups and higher {Bruhat}
  orders.
\newblock In {\em Algebraic number theory -- in honor of {K}. {Iwasawa}},
  volume~17 of {\em Adv. Stud. Pure Math.} Academic Press, 1989.

\bibitem{szaboWelzl01}
Tibor Szabó and Emo Welzl.
\newblock Unique sink orientations of cubes.
\newblock In {\em Proc. 42nd IEEE Symposium on Foundations of Computer
  Science}, pages 547--555, 2001.
\newblock \href {https://doi.org/10.1109/SFCS.2001.959931}
  {\path{doi:10.1109/SFCS.2001.959931}}.

\end{thebibliography}
